\documentclass{article}
\usepackage[margin=1.2in]{geometry}

\usepackage{authblk}
\usepackage{graphicx,subfigure}
\usepackage{enumerate}
\usepackage{verbatim}
\usepackage{amssymb,amsmath,amsthm}
\usepackage{color}
\usepackage{hyperref}
\usepackage{cite}
\usepackage{xspace}
\usepackage{microtype} 

\DeclareGraphicsExtensions{.pdf}

\newtheorem{theorem}{Theorem}
\newtheorem{lemma}[theorem]{Lemma}
\newtheorem{corollary}[theorem]{Corollary}

\newtheorem{definition}[theorem]{Definition}

\newcommand{\Int}{{\mathbb Z}}

\newcommand{\outdeg}{\operatorname{out-deg}}

\newcommand{\WLOG}{without loss of generality \xspace} 

\newcounter{dummycount}

\newcommand{\wormhole}[1]
{
\newcounter{#1}
\setcounter{#1}{\value{theorem}}
}

\newenvironment{backInTime}[1]
{
\setcounter{dummycount}{\value{theorem}}
\setcounter{theorem}{\value{#1}}
}
{
\setcounter{theorem}{\value{dummycount}}
}

\newcommand{\df}{\textit}

\title{Contact Representations of Sparse Planar Graphs}

\author[1]{Md.~Jawaherul~Alam}
\author[2]{David~Eppstein}
\author[3]{Michael~Kaufmann}
\author[1]{Stephen~G.~Kobourov}
\author[1]{Sergey~Pupyrev}
\author[4]{Andr\'e~Schulz}
\author[5]{Torsten~Ueckerdt}
\affil[1]{Department of Computer Science, University of Arizona, USA
}
\affil[2]{Computer Science Department, University of California, Irvine, USA
}
\affil[3]{Wilhelm-Schickard-Institut f\"ur Informatik, Universit\"at T\"ubingen, Germany
}
\affil[4]{Institut Math.~Logik \& Grundlagenforschung, Universit\"at M\"unster, Germany
}
\affil[5]{Department of Mathematics, Karlsruhe Institute of Technology, Germany
}

\begin{document}
\date{}

\maketitle

\begin{abstract}
We study representations of graphs by \textit{contacts of circular arcs}, \textit{CCA-representations}
 for short, where the vertices are interior-disjoint circular arcs in the plane and each edge is realized
 by an endpoint of one arc touching the interior of another. A graph is \textit{$(2, k)$-sparse}
 if every $s$-vertex subgraph has at most $2s-k$ edges, and \textit{$(2,k)$-tight} if
 in addition it has exactly $2n-k$ edges, where $n$ is the number of vertices.
 Every graph with a CCA-representation is planar
 and $(2,0)$-sparse, and it follows from known results on contacts of line segments that for $k\ge 3$ every $(2,k)$-sparse graph has a
 CCA-representation. Hence the question of CCA-representability is open for $(2,k)$-sparse graphs with $0\le k\le 2$. We partially answer this question by computing CCA-representations for several
 subclasses of planar $(2, 0)$-sparse graphs. In particular, we show that every plane $(2, 2)$-sparse
 graph has a CCA-representation, and that any plane $(2,1)$-tight graph or $(2,0)$-tight graph dual to a $(2,3)$-tight graph or $(2,4)$-tight graph has a CCA-representation.

Next, we study CCA-representations in which each arc has an empty convex hull. We characterize the plane graphs that have
such a representation, based on the existence of a special orientation of the graph edges.
Using this characterization, we show that every plane graph of maximum degree 4 has such a representation, but that finding such a representation for a plane $(2,0)$-tight graph
with maximum degree $5$ is an NP-complete problem.

Finally, we describe a simple algorithm for representing plane $(2, 0)$-sparse graphs with \textit{wedges},
 where each vertex is represented with a sequence of two circular arcs (straight-line segments).
\end{abstract}

\newpage


\begin{figure}[b]
\vspace{-2ex}
\centering
\includegraphics[width=0.22\textwidth]{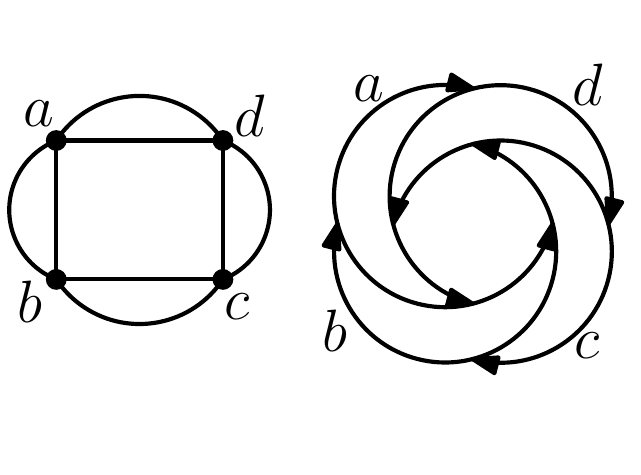}
\hspace{0.1cm}
\includegraphics[width=0.37\textwidth]{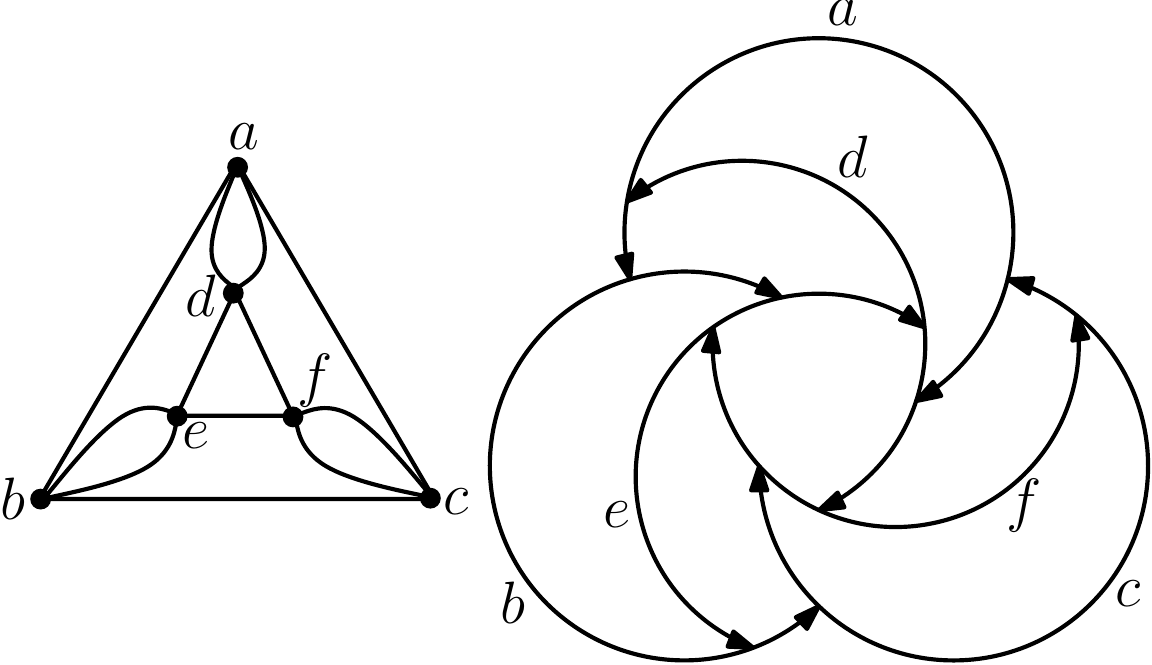}
\hspace{0.1cm}
\includegraphics[width=0.35\textwidth]{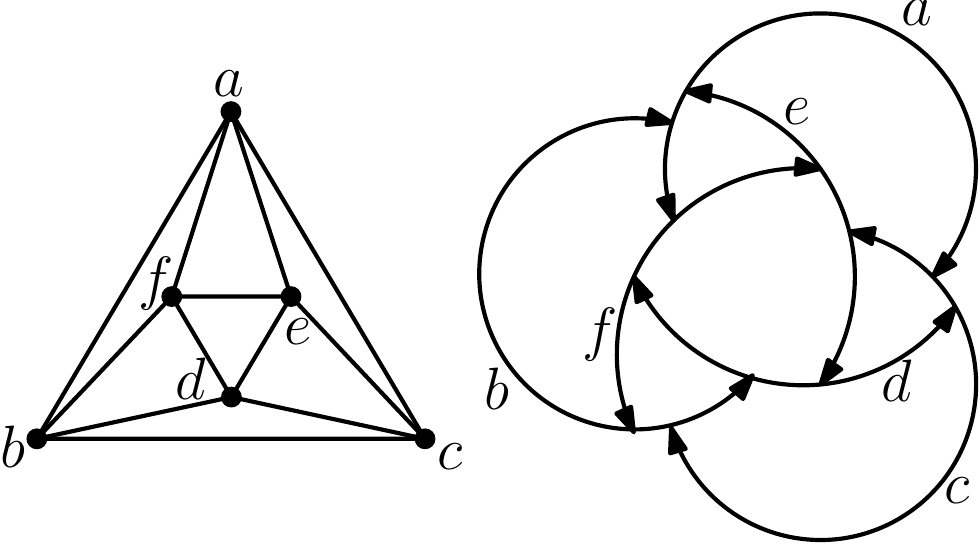}\\
(a)\hspace{0.3\textwidth}(b)\hspace{0.35\textwidth}(c)\hspace{0.1\textwidth}
\caption{Circular-arc contact representations of two multigraphs (a)--(b) and a simple graph (c).}
\label{fig:teaser}
\vspace{1ex}
\end{figure}

\section{Introduction}
In a \df{contact representation} of a planar graph, the vertices are represented by non-overlapping geometric objects such as circles, polygons, or line segments and the edges are realized by a prespecified type of contact between these objects. Contact graphs of circles, made famous by the Koebe--Andreev--Thurston circle packing theorem~\cite{Koebe36}, have many applications in graph drawing~\cite{Malitz1994,BerEpp-WADS-01,Moh-GD-99,KesPacPal-SJDM-13,Rot-GD-11,AicRotSch-CGTA-12,BekRaf-GD-12,Epp-GD-12,EppHolLof-GD-13},
and this success has motivated the study of many other contact representations~\cite{FMR94,homothetic07,GonLevPin-DCG-12,ghkk10}.
The special cases of contact representations with curves and in particular with line segments are of particular interest~\cite{FMP91,Hli98,Hli01,CCD04,FM07-segment}.
We consider a novel type of contact representation where a vertex is represented by a circular arc, and an edge corresponds to one arc touching an interior point of another arc; see Fig.~\ref{fig:teaser}. Note that we do not count tangencies between interior points as contacts as this would trivialize the problem. These representations are a generalization of the contacts of straight-line segments (which may be thought of as circular arcs with infinite radius) and we call them {\em contacts of circular arcs} or {\em CCA-representations} for short.

Every $k$-vertex induced subgraph of a contact graph of curves in the plane has at most $2k$ edges, because every edge uses up one of the curve endpoints. This motivates us to study classes of sparse graphs defined by limits on the numbers of edges in their subgraphs.
A graph $G = (V,E)$ is said to be \df{$(p,k)$-sparse}~\cite{Lee08} if for every $W \subseteq V$ we have
\[
 |E[W]| \leq \max\{p|W| - k, |W|-1\};
\]
it is \df{$(p,k)$-tight} if in addition $|E| = p|V| - k$. For example, $(1,1)$-sparse graphs are exactly forests, while $(1,1)$-tight graphs are trees, and the observation above can be rephrased as stating that all graphs representable by circular arcs are $(2,0)$-sparse.
This definition makes sense only for $k<2p$: for larger $k$, each two vertices would induce no edges and the graph would be empty. However,
we may extend the definition by restricting $|W|$ to be larger than two. Thus, we define a graph to be $(2,4)$-sparse if every $s$-vertex subgraph with $s\ge 3$ has at most $2s-4$ edges, and $(2,4)$-tight if in addition it has exactly $2|V|-4$ edges. A planar graph is $(2,4)$-sparse if and only if it is triangle-free and $(2,4)$-tight if and only if it is a maximal bipartite planar graph.
The same idea can be extended to larger $k$ by restricting $|W|$ to be even greater, but in the remainder we consider only planar $(2,k)$-sparse and planar $(2,k)$-tight graphs for $k \in \{0,1,2,3,4\}$.

A graph admits a curve contact representation if and only if it is planar (because the curves do not cross) and $(2,0)$-sparse (each subset of $s$ curves has at most $2s$ contacts)~\cite{Hli98}.
On the other hand, a planar graph has a contact representation with line segments if and only if it is $(2,3)$-sparse~\cite{Tho93,ourJGAA12}.
Hence, natural questions arise: What are the simplest curves that can represent all planar $(2,0)$-sparse graphs? Perhaps the simplest non-straight curves are circular arcs, so how powerful are circular arcs in terms of contact representations? In particular, does every planar $(2,k)$-sparse graph
have a CCA-representation for $k \in \{0,1,2\}$? We partially answer these questions by computing circular-arc contact representation for several subclasses
of planar $(2,0)$-sparse graphs, and by finding a $(2,0)$-sparse plane multigraph that does not have such a representation.

As another contribution, we resolve the open problem stated by de Fraysseix and Ossona de Mendez~\cite{FM07-segment}. They pointed out that
any contact representation with curves is homeomorphic to a contact representation with polylines composed of three segments, and
asked whether two segments per polyline is sufficient. We affirmatively answer the question.

\medskip
\noindent{\bfseries Our Results.}
We first need some structural results about sparse planar graphs.

\begin{lemma}[Augmentations]
\label{lem:augmentation}
~\\[-3ex]
\begin{itemize}
  \item For every integer $k \in \{0,1,2,3\}$, every plane $(2,k)$-sparse graph is a
 spanning subgraph of some plane $(2,k)$-tight graph.

  \item For $k = 4$, a $(2,k)$-sparse graph forms a subgraph of a $(2,k)$-tight graph if and only
	if it is bipartite. In particular, the $5$-cycle is a $(2,4)$-sparse graph that is not a subgraph
	of a $(2,4)$-tight graph.
\end{itemize}
\end{lemma}

\begin{lemma}[Plane Duals]
\label{lem:duals}
~\\[-3ex]
 \begin{itemize}
  \item For every $k \in \{0,1\}$ and every integer $\ell \in \Int$, there is a plane $(2,k)$-tight graph
	whose dual is not $(2,\ell)$-sparse.

  \item For every $k \in \{2,3,4\}$, the dual of every plane $(2,k)$-tight graph is $(2,4-k)$-tight.
 \end{itemize}
\end{lemma}

Lemma~\ref{lem:duals} implies that for $k=2$, duality is an involution on the plane $(2,4-k)$-tight graphs, and therefore, every plane $(2,4-k)$-tight graph can be obtained as the dual of a plane $(2,k)$-tight graph. However, for $k\in\{3,4\}$, the duals of plane $(2,k)$-tight graphs form a proper subclass of all plane $(2,4-k)$-tight graphs. In fact, we prove that the dual of every plane $(2,3)$-tight graph is a co-Laman graph
(a graph in which $|E| = 2|V|-1$ and $E[W] \leq 2|W| - 2$ for all $W \subsetneq V$; see Fekete~et al.~\cite{FJW04}).
And one can easily prove that the dual of every $(2,4)$-tight graph is
 $4$-regular.

Our main results are:

\begin{theorem}[Contact Circular Arc Representations]\label{thm:CCA}
~\\[-3ex]
 \begin{itemize}
  \item Every plane $(2,2)$-sparse graph has a contact representation with circular arcs.
  \item Every plane co-Laman multigraph has a contact representation with circular arcs.
  \item Every plane graph with maximum degree $4$ has a contact representation with circular arcs.
  \item There is a plane $(2,0)$-tight multigraph with no contact representation with circular arcs.
 \end{itemize}
\end{theorem}

The theorem above directly implies the following corollary.
\begin{corollary}
 For every $k \in \{0,1,2,3,4\}$ and every plane $(2,k)$-tight graph $G$, the plane dual $G^*$ of $G$ has a contact representation with circular arcs, whenever $G^*$ is $(2,4-k)$-tight.
\end{corollary}

We use two different approaches to construct CCA representations.
The first approach is a constructive one, and it can be used for plane $(2,2)$-tight graphs and plane co-Laman graphs. We find a special construction sequence for each graph in these two classes (similar to the Henneberg moves~\cite{Hen11} for $(2,3)$-tight graphs, also see~\cite{Whiteley96,FJW04,Nix11}), and we show that this construction sequence can be modified into a construction sequence for a CCA-representation. The second approach is structural.
For some  planar $(2,0)$-sparse graphs, in particular for all graphs of maximum degree $4$,
 we can obtain a stronger form of contact representation where the convex hull of each arc is empty.
For these graphs we define the notion of a \emph{good $2$-orientation}, and use a circle packing construction to find this stronger CCA-representation from the good orientation. However, we show

\begin{theorem}
\label{thm:empty-hull-NPC}
Testing whether a planar $(2,0)$-tight graph has a contact representation where the convex hull of each arc is empty is NP-complete, even for graphs of maximum degree $5$.
\end{theorem}

Finally we consider contact representation with \emph{wedges} (that is, polyline segments with at most
one bend). A wedge can be viewed as a sequence of two circular arcs (straight-line segments).
It is not difficult to prove that every planar $(2,0)$-sparse graph has a contact representation with polylines composed of three
segments~\cite{FM07-segment}. On the other hand, as pointed out earlier, one segment per polyline is not sufficient. This raises a question
(asked in~\cite{FM07-segment}) whether every planar $(2,0)$-sparse graph has a contact representation with polylines composed of two segments.
We resolve the question by showing that every plane $(2,0)$-sparse graph has a contact representation with wedges.

\wormhole{thm-wedge}
\begin{theorem}
\label{thm:wedge-representation}
 Every plane $(2,0)$-sparse graph has a contact representation where each vertex is represented by a wedge.
\end{theorem}

\section{Preliminaries}\label{sec:preliminaries}

Here we give some preliminary results on $(2,k)$-tight graphs and prove
 Lemmas~\ref{lem:augmentation}~and~\ref{lem:duals}.
We start by listing some known characterizations of $(2,k)$-tight graphs for $k=2,3$.

\begin{lemma}\label{lem:disjoint-trees}
 A graph $G$ is $(2,2)$-tight if and only if $G$ is the union of two edge-disjoint spanning trees (Tay~\cite{Tay91}; see also~\cite{Nixon-phd,NW61,Tut61}).
$G$ is $(2,3)$-tight if and only if for every edge $e$ of $G$ the graph obtained by doubling $e$ is $(2,2)$-tight. (Graver~et al.~\cite{GSS93})
\end{lemma}
%
%

An \df{orientation} of a graph $G$ is a directed graph whose underlying undirected graph is~$G$. We call it a \df{$2$-orientation} if every vertex has out-degree exactly $2$, and a \df{$2^-$-orientation} if every vertex has out-degree at most $2$.

Lemma~\ref{lem:orientations} characterizes $(2,i)$-tight graphs in terms of their orientations.

\wormhole{lem:orient}
\begin{lemma}
\label{lem:orientations}
 Let $G = (V,E)$ be a graph. Then each of the following holds.
~\\[-4ex]
 \begin{enumerate}[\rm (i)]
  \item $G$ is $(2,0)$-tight if and only if $G$ admits a $2$-orientation.\label{enum:orient-20}

  \item $G$ is $(2,1)$-tight if and only if $G$ plus a loop at an arbitrary vertex $v \in V$ admits a $2$-orientation in which for every $u \in V$ there exists a directed $u$-to-$v$ path.\label{enum:orient-21}

  \item $G$ is $(2,2)$-tight if and only if $G$ after adding two loops at an arbitrary vertex $v \in V$ has a $2$-orientation where for every $u \in V$ there are two edge-disjoint directed $u$-to-$v$ paths.\label{enum:orient-22}

  \item If $G$ is $(2,k)$-sparse for some $k \geq 0$, then $G$ admits a $2^-$-orientation.\label{enum:orient-2k}
 \end{enumerate}
\end{lemma}

\begin{proof}
We prove this lemma, following the pebbling game of Lee and Streinu~\cite{Lee08}, according to which a graph is $(2,k)$-sparse (for $0\leq k< 4$) if and only if it is possible to construct it, starting from a graph with no edges and two pebbles per vertex, by a sequence of the following operations:
\begin{enumerate}[(O1)]
 \item Add an edge between two vertices that together have at least $k+1$ pebbles while no vertex has more than two pebbles.
 Remove one of the pebbles and direct the edge away from the vertex from which it was removed.\label{enum:pebble-1}
 \item Move a pebble backwards along a directed path in the graph, to a vertex that did not already have two pebbles, reversing the orientation of all the edges it moves along.\label{enum:pebble-2}
\end{enumerate}
An easy induction shows that at each step of this process each vertex with $i$ pebbles on it has exactly $2-i$ outgoing edges. Note that case~(\ref{enum:orient-2k}) of the lemma follows immediately, since every $(2,k)$-sparse graph is also $(2,0)$-sparse. We now prove for the remaining three cases.

 For case~(\ref{enum:orient-20}), observe that each operation~(O\ref{enum:pebble-1}) uses up one of the initial $2|V|$ pebbles, so a $(2,k)$-tight graph must have exactly $k$ pebbles remaining at the end of the process. In particular a $(2,0)$-tight graph has no pebbles and each vertex has exactly $2$ outgoing edges.

 For case~\ref{enum:orient-21}, the oriented graph resulting from the pebbling game will have a pebble on exactly one vertex $v$. There can be no subset of vertices that does not contain $v$ and has no outgoing edges, because such a subset would induce a $(2,0)$-tight subgraph, so every vertex has a path to $v$. Adding a self-loop at $v$ gives the desired structure. The vertex $v$ can be chosen arbitrarily, because the pebble can be moved backwards along the path that connects any other vertex to $v$.

 For case~\ref{enum:orient-22}, the pebbling game leaves two pebbles on vertices. There can be no subset of vertices for which the number of pebbles in the subset plus the number of edges escaping the subset is less than two, because such a subset would induce a $(2,0)$-tight or $(2,1)$-tight subgraph. Therefore, by similar reasoning to case~\ref{enum:orient-21}, these two pebbles can both be moved to an arbitrarily chosen vertex $v$, after which every other vertex must have two edge-disjoint oriented paths to $v$.
\end{proof}


Next we prove Lemma~\ref{lem:augmentation}, that for $k \in \{0,1,2,3\}$ every plane $(2,k)$-sparse graph can be augmented by adding edges into a plane $(2,k)$-tight graph. For $k=4$ such an augmentation exists if and only if the graph is bipartite. We distinguish the five cases for $k$.

\begin{description}
 \item[Case~1, $k = 0$.] If we allow multigraphs, then it is easy to augment $G$.
By Lemma~\ref{lem:orientations}, $G$ has an orientation with $\outdeg(v)\le 2$ for each $v\in V$.
 If $\outdeg(v) = 2$ for each $v \in V$, $G$ is already $(2,0)$-tight. Otherwise, take a vertex $v$ with $\outdeg(v) \leq 1$ and double the outgoing edge incident to $v$ to increase the out-degree without breaking planarity.

We prove the stronger claim
that any plane $(2,0)$-sparse simple graph with $n\ge 6$
 vertices can be augmented to a plane $(2,0)$-tight simple graph. Assume for a
 contradiction that there is a plane $(2,0)$-sparse simple graph with $n\ge 6$ vertices
 that cannot be augmented in this way, and let $G$ be such a graph
 with the maximum number of edges. Since $G$ is not $(2,0)$-tight and $n\ge 6$, $G$ has
 a non-triangular face. If $G$ has no $(2,0)$-tight subgraph, then we can add an edge
 inside such a face without violating $(2,0)$-sparseness and planarity, contradicting
 the maximality of $G$. Otherwise, there is an inclusion-maximal $(2,0)$-tight subgraph
 $H$ of $G$. Choose an edge $e = (u,v)$ in $H$ and a vertex $w$ not
 in $H$, on the same face of $G$. Since $H$ is inclusion-maximal, at least one of
 the edges $(u,w)$ and $(v,w)$ is not in $G$ and can be added, again
 contradicting the maximality of $G$.

\item[Case~2, $k = 1$.] We play the pebble game of Lee and Streinu~\cite{Lee08} for the $(2,1)$-sparse
 graph and end with a directed graph $D=(V,\bar E)$ whose vertices might contain pebbles.
 Let $x\in V$ be a vertex with a pebble and let $V_x$ be the set of vertices that have a path to $x$ in $D$. If two sets $V_x$ and $V_y$ intersect (with $x\not =y$) we can move two pebbles to the endpoint of an edge and then add a copy of this edge while following the pebble game rules, giving a valid augmentation of the graph. Otherwise, for each $x\not =y$, the sets $V_x$ and $V_y$ are disjoint.
 We can further assume that there is no edge between two different vertex sets $V_x$ and $V_y$ since this edge would extend one of the sets.

 Now let $\bar V$ denote all vertices that are in none of the sets $V_x$. We have to rule out the case that $D$ restricted to $\bar V$ ``separates'' the sets $V_x$. Assume that $\bar V \not = \emptyset$. Let $e$ be the edge last added with both endpoints in $\bar V$. To add $e$, we needed two pebbles. One pebble remained at some vertex $u\in \bar V$ after the addition of $e$. Since $\bar V$ has no pebbles, this pebble must have been removed from $\bar V$, leaving a path that connects $u$ to the vertex where the pebble remains, or where it was consumed. This place lies outside of $\bar V$, say in $V_x$. This shows that $u\in V_x$,  a contradiction.

 The augmentation can now be carried out as follows. The graph $D$ decomposes into connected components, associated with the sets $V_x$. Select two components, move the free pebbles to the boundaries, and add an edge that connects the two pebbled vertices.

 \item[Case~3, $k = 2$.] By Lemma~\ref{lem:disjoint-trees}, a graph is $(2,2)$-tight if and only if its edges can be partitioned into two spanning trees. A partition of the edges into two forests $F_1$ and $F_2$ can be obtained by several known algorithms,~\cite{RT85} for example. Next find a face $f$ and two vertices $u$, $v$ on $f$ that are not in the same component of one of the forests, say $F_1$. If $f$ has length at least 4, one can add an edge into it and put this edge into $F_1$ so that two components of $F_1$ merge. Otherwise $f=uvw$ is a triangle and exactly two edges of $f$ are in $F_2$ and one in $F_1$, say $vw$. Now we consider the face $f'$ on the other side of the edge $uw$ and repeat the argument until we find a face of length at least $4$ this way. We find some face of length at least $4$ at some point; since otherwise we have a cycle of faces, which gives a cycle in $F_1$ or $F_2$, a contradiction. Repeated insertion of edges in this way eventually turns both $F_1$ and $F_2$ into trees and hence $G$ becomes $(2,2)$-tight.

 \item[Case~4, $k = 3$.] If $G=(V,E)$ is $(2,3)$-sparse but has fewer than $2|V|-3$ edges, then it represents a flexible bar-joint framework, with at least one nontrivial motion, and a face $f$ that is not rigid during this motion. Let $x$ be a vertex on $f$ that changes its angle during the motion. Further, let $v_1,v_2,\ldots ,v_k$ be the vertices of $f$ that are visible from $x$ in circular order. There must exist two vertices $v_i$ and $v_{i+1}$ for which the angle $\angle v_i x v_{i+1}$ changes. If edge $e= (v_i,v_{i+1})$ is not in $G$ we add $e$ to $G$. If $e$ already exists, at least one distance between $x$ and $v_{i}$ or between $x$ and $v_{i+1}$ changes during the motion; we add the corresponding edge to $G$. In all cases we added an edge with endpoints in different rigid components of $G$, eliminating one degree of freedom from the space of motions of $G$. Hence, $G$ is a Laman graph or a subgraph of a Laman graph and remains $(2,3)$-sparse. Continue to add edges, until the graph is $(2,3)$-tight. See~\cite{ourJGAA12}
for an alternate proof.


 \item[Case~5, $k=4$.] If a $(2,4)$-sparse planar graph is not bipartite, then it remains non-bipartite no matter what is added to it. Thus, it cannot be a subgraph of a maximal bipartite planar graph, which is the same thing as a $(2,4)$-tight planar graph. On the other hand, if a $(2,4)$-sparse graph is bipartite but not tight, then in any planar embedding it has a face with six or more vertices. In such a face, one can find three mutually-crossing properly colored diagonals. It is not possible for all three to already be part of the graph, because this would cause the graph to contain a $K_{3,3}$ minor. Therefore, one of these diagonals can be added, splitting the face while preserving bipartiteness. Repeating this operation leads to a graph in which all faces are quadrilaterals, which is $(2,4)$-tight.
\end{description}


To conclude our preliminaries
we prove Lemma~\ref{lem:duals}, relating the sparseness parameters of a graph and its dual.

First, let $\ell \in \Int$ be any integer. We shall construct a plane $(2,0)$-tight graph $G_0$ and a plane $(2,1)$-tight graph $G_1$ whose duals $G_0^*$ and $G_1^*$ contain two vertices $v$ and $w$ that are joined by $5-\ell$ parallel edges. In particular, $G_0^*$ and $G_1^*$ will not be $(2,\ell)$-sparse.

\begin{figure}[tb]
 \centering
 \subfigure[\label{fig:dual-counterexample}]{
  \includegraphics[scale=0.8]{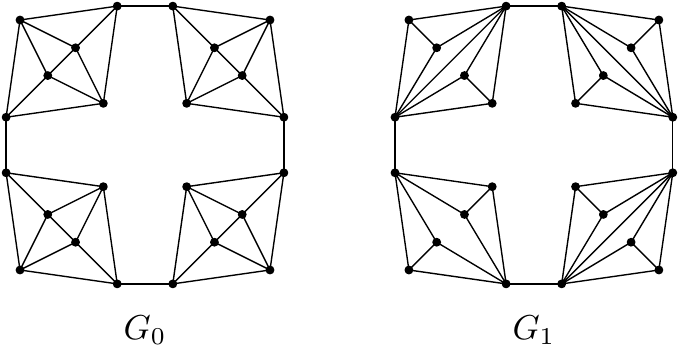}
 }
 \hspace{1em}
 \subfigure[\label{fig:segment-contact-dual}]{
  \includegraphics[scale=0.75]{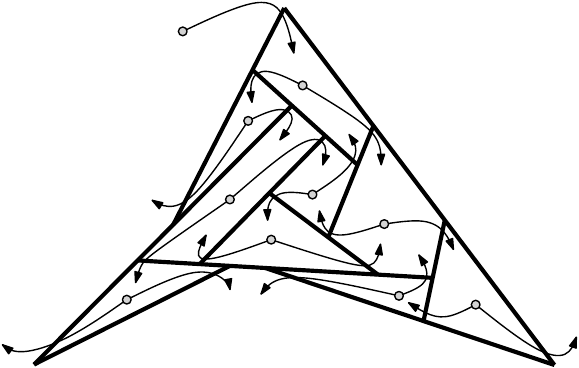}
 }
 \caption{(a)~A plane $(2,0)$-tight graph $G_0$ and a plane $(2,1)$-tight graph $G_1$ whose duals are not $(2,1)$-sparse. (b)~A segment contact representation of a plane $(2,3)$-tight graph and the corresponding orientation of the dual.}
\end{figure}

We start with a cycle $C$ of order $2(5 -\ell)$ and replace every other edge $e$ of $C$ with a certain graph $H$. Let $H_0$ and $H_1$ be any plane $(2,0)$-tight graph and any plane $(2,1)$-tight graph, respectively, and $e_0$, $e_1$ be any edge in $H_0$, $H_1$, respectively. For $G_0$, we replace every other edge $e$ of $C$ with $H_0 - e_0$, where the endpoints of $e_0$ and $e$ are identified. For $G_1$, we replace every other edge $e$ of $C$ with $H_1$, where the endpoints of $e_1$ and $e$ are identified, and in one such copy of $H_1$ we remove the edge $e_1$. It is easy to see that the resulting graphs $G_0$ and $G_1$ have the desired properties. We refer to Fig.~\ref{fig:dual-counterexample} for an illustration.

\bigskip

\noindent
Now, let $k \in \{2,3,4\}$ and $G = (V,E)$ be a planar $(2,k)$-tight graph with a fixed plane embedding. We shall show that the dual $G^*$ of $G$ is $(2,4-k)$-tight. By Euler's formula,
\[
 |E(G^*)| = |E(G)| = 2|V(G)| - k = 2(|E(G^*)| - |V(G^*)| + 2) - k
\]
and thus $|E(G^*)| = 2|V(G^*)| - (4-k)$. So it remains to show that $G^*$ is $(2,4-k)$-sparse. We distinguish the three cases for $k$.

\begin{description}
 \item[Case~1, $k = 4$.] As $G$ is $(2,4)$-sparse, it is triangle-free, and as $|E(G)| = 2|V(G)| - 4$, $G$ is a plane quadrangulation. In particular, all faces of $G$ are quadrangular, which means that $G^*$ is $4$-regular and thus $G^*$ is $(2,0)$-tight.

 \item[Case~2, $k = 3$.] $G$ is $(2,3)$-tight, so by Lemma~\ref{lem:disjoint-trees} it has a decomposition into two edge-disjoint spanning trees after doubling any edge~$e$. A spanning tree $T$ in a plane graph corresponds to a spanning tree in its dual graph $G^*$, formed by the dual edges of all edges not in $T$. Hence, $G^*$ can be decomposed into two spanning trees after deleting any edge. Let $v$ be an arbitrary vertex in $G^*$. We delete some edge $e^*=(u,v)$ in $G^*$, fix two edge-disjoint spanning trees of $G^*\setminus e^*$, and orient each edge towards the root $v$. We end up with a $2$-orientation that meets the requirements of Lemma~\ref{lem:orientations}(\ref{enum:orient-22}), and therefore with a $(2,2)$-tight graph. Adding back $e^*$ results in a $(2,1)$-tight graph.

 We remark that there is an alternative, more geometric proof for this case. As $G$ is a planar $(2,3)$-tight graph, $G$ admits a segment contact representation~\cite{FM07-segment}. Every vertex $v$ of $G^*$, except the vertex $v_o$ for the outer face of $G$, corresponds to a convex polygon $P(v)$ in this representation; see Fig.~\ref{fig:segment-contact-dual}. Now, we add in $G^*$ a loop at $v_o$ and obtain a $2$-orientation of this graph by orienting outgoing for every $v \neq v_o$ the edges corresponding to the leftmost and rightmost point of $P(v)$, and the remaining two edge outgoing at $v_o$.

 \item[Case~3, $k = 2$.] As $G$ is $(2,2)$-tight, by Lemma~\ref{lem:disjoint-trees}, $G$ is the union of two edge-disjoint spanning trees. Since $G^*$ is connected, for every spanning tree $T$ of $G$, the edges of $G^*$ that are not crossed by any edge in $T$ form a spanning tree of $G^*$. Hence, $G^*$ is the edge-disjoint union of two spanning trees. Thus, again by Lemma~\ref{lem:disjoint-trees}, $G^*$ is $(2,2)$-tight.
\end{description}

\section{Contact Representations from Henneberg Moves}
\label{sec:CCA}

Here we prove the existence of CCA-representations for $(2,2)$-sparse and co-Laman graphs,
 the first two cases of Theorem~\ref{thm:CCA}. We defer the degree-4 and $(2,0)$-tight cases
 to Section~\ref{sec:good}.

We begin by describing a set of moves which can be applied to a plane $(2,k)$-tight graph, in order to obtain a larger plane $(2,k)$-tight graph (with more vertices), where $k \in \{0,1,2,3,4\}$ depends on the type of move. Afterwards we show that certain subsets of these moves can be used to generate all plane $(2,k)$-tight graphs of a certain class of graphs, starting from one concrete base graph. All but one of these moves are well-known and have already successfully been used for this purpose; see Fig.~\ref{fig:moves}.

\begin{figure}[t]
  \centering
  \includegraphics[width=\textwidth]{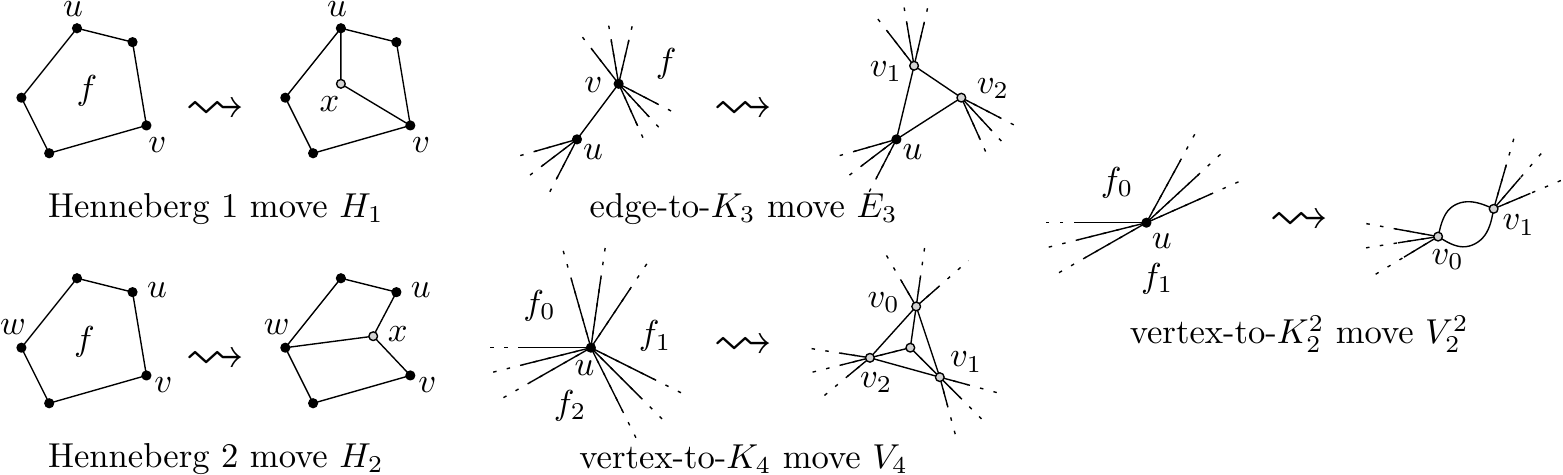}
  \caption{The moves $H_1$, $H_2$, $E_3$, $V_4$ and $V_2^2$.}
  \label{fig:moves}
 \end{figure}

\begin{definition}[Moves]
 Let $G = (V,E)$ be a plane $(2,k)$-graph for some $k \in \{0,1,2,3,4\}$.
 \begin{description}
  \item[The Henneberg~1 move $H_1$.] For a face $f$ of $G$ and two distinct vertices $u,v$ on $f$, introduce a new vertex $x$ inside $f$ and add edges from $x$ to $u$ and $v$.

  \item[The Henneberg~2 move $H_2$.] For a face $f$ of $G$ and an edge $e = (u,v)$ on $f$ and a third vertex $w \neq u,v$ on $f$, introduce a new vertex $x$ inside $f$, add edges from $x$ to $u$, $v$ and $w$, and remove the edge $e$.

  \item[The edge-to-$K_3$ move $E_3$.] For an edge $e = (u,v)$ of $G$ and a face $f$ incident to $v$, replace $v$ by two vertices $v_1,v_2$ connected by an edge $(v_1,v_2)$, and add edges from $v_1$ ($v_2$) to each neighbor of $v$ that lies clockwise (counterclockwise) between $f$ and $e$ (included) around $v$.

  \item[The vertex-to-$K_4$ move $V_4$.] For a vertex $u$ of $G$ and three (not necessarily distinct) faces $f_0,f_1,f_2$ incident to $u$, appearing in that clockwise order around $u$, replace $u$ by a plane $K_4$ with outer vertices $v_0,v_1,v_2$, and add edges from $v_i$ to every neighbor of $u$ that lies clockwise between $f_i$ and $f_{i+1}$ around $u$, $i=0,1,2$, where indices are taken modulo $3$.

  \item[The vertex-to-$K_2^2$ move $V_2^2$.] For a vertex $u$ of $G$ and two (not necessarily distinct) faces $f_0,f_1$ at $u$, replace $u$ by two vertices $v_1,v_2$ connected by two parallel edges, and add edges from $v_i$ to every neighbor of $u$ that lies clockwise between $f_i$ and $f_{i+1}$ around $u$, $i=0,1$, where indices are taken modulo $2$.
 \end{description}
\end{definition}

Henneberg moves $H_1$ and $H_2$ were introduced by Henneberg~\cite{Hen11}, moves $E_3$ and $V_2^2$ were defined by Whiteley~\cite{Whiteley96}, $E_3$ also appears in~\cite{FJW04} under the name \df{vertex-splitting}, and the move $V_4$ was introduced by Nixon and Owen~\cite{Nix11}.

The first part of the following Lemma is due to Henneberg~\cite{Hen11}, see also Lov\'asz and Yemini~\cite{Lov82} and Haas~et al.~\cite{HOR05}.

\begin{lemma}\label{lem:moves-generate}
 Each of the following holds.
~\\[-4ex]
 \begin{enumerate}[\rm (i)]
  \item All plane $(2,3)$-tight graphs can be generated by $H_1$ and $H_2$ moves starting from a triangle~\cite{Hen11,Lov82,HOR05}.\label{enum:generate-23}
  \item All duals of plane $(2,3)$-tight graphs can be generated by $E_3$ and $V_2^2$ moves starting from three parallel edges.\label{enum:generate-23-dual}
  \item All plane $(2,2)$-tight graphs can be generated by $E_3$ and $V_4$ moves starting from an isolated vertex.\label{enum:generate-22}
 \end{enumerate}
\end{lemma}

In order to prove (\ref{enum:generate-22}) we need one more concept from the literature. A \df{Laman-plus-one} graph is a simple graph $G$ with an edge $e = uv$, such that $\deg(u)\geq 2$ and $\deg(v) \geq 2$, and $G-e$ is a $(2,3)$-tight graph (Laman graph). Laman-plus-one graphs form a proper subclass of $(2,2)$-tight graphs. Fekete~et~al.~\cite{FJW04} had the following claim without a proof about generating Laman-plus-one graphs by $E_3$ moves starting from $K_4$.
 For the sake of completeness here we give a proof for this claim.

\wormhole{lem:lam11}
\begin{lemma}[Fekete~et al.~\cite{FJW04}]
\label{lem:lam+1}
Every plane Laman-plus-one graph can be generated by $E_3$ moves starting from $K_4$.
\end{lemma}

In order to prove Lemma~\ref{lem:lam+1}, we need the following three auxiliary lemmas.
 The first of these lemmas is due to~\cite{FJW04} and can be easily proved by counting.

\begin{lemma}
\label{lem:critical} Let $G$ be a $(2,3)$-tight graph and let $H_1$ and $H_2$ be two
 subgraphs of $G$ which are both $(2,3)$-tight such that $|V(H_1)\cap V(H_2)|\ge 2$.
 Then $H_1\cup H_2$ is also $(2,3)$-tight.
\end{lemma}

 For an edge $e$ of $G$, denote by $G/e$ the graph obtained from $G$ after contracting
 edge $e$. If $G$ is a plane graph, then the planar embedding of $G$ also defines a
 planar embedding for $G/e$. Let $G$ be plane $(2,3)$-tight graph and let $e$ be an edge
 of $G$. Then $e$ is called \df{contractible} if the graph $G/e$ is also a plane $(2,3)$-tight
 graph. One can see that if $G$ has a contractible edge $e$, then $G$ can be generated from
 $G/e$ by an $E_3$ move. The following lemma is due to~\cite{FJW04}.

\begin{lemma}[Fekete~et al.~\cite{FJW04}]
\label{lem:blocker}
Let $G$ be a $(2,3)$-tight graph with $|V(G)|\ge 4$ and let $uvw$ be a triangle of $G$
 such that the edge $uv$ is not contractible. Then there is a proper $(2,3)$-tight subgraph
 $H$ of $G$ such that $u,v\in V(H)$ and $w\notin V(H)$. Furthermore, there is some edge
 $e'\in E(H)$ which is contractible in $G$.
\end{lemma}

We also have the following lemma.

\begin{lemma}
\label{lem:lam+1-K4} Let $G$ be a plane graph with a $K_4$ as a subgraph $H$ formed by four
 vertices $u,v,x,y$. If $G-(u,v)$ is a $(2,3)$-tight graph, then $G$ can be generated by $E_3$
 moves starting from $H$.
\end{lemma}
\begin{proof} Consider the four triangle faces of $H$. For each of these faces $f$, denote
 $U_f$ the set of vertices inside the face $f$. We first claim that for each of these four faces
 $f$, the subgraph of $G$ induced by the vertices $V(f)\cup U_f$ is a $(2,3)$-tight graph;
 where $V(f)$ denotes the vertices on the cycle $f$. Indeed, if $f$ does not contain the edge
 $(u,v)$ (assume $f=vxy$ since the case when $f=uxy$ is similar), then the claim follows
 from Lemma~\ref{lem:tight-component}, since both $G-(u,v)$ and $f=vxy$ are $(2,3)$-tight;
 see Fig.~\ref{fig:lam+1-K4}(a).

\begin{figure}[h]
\centering
\includegraphics[width=0.7\textwidth]{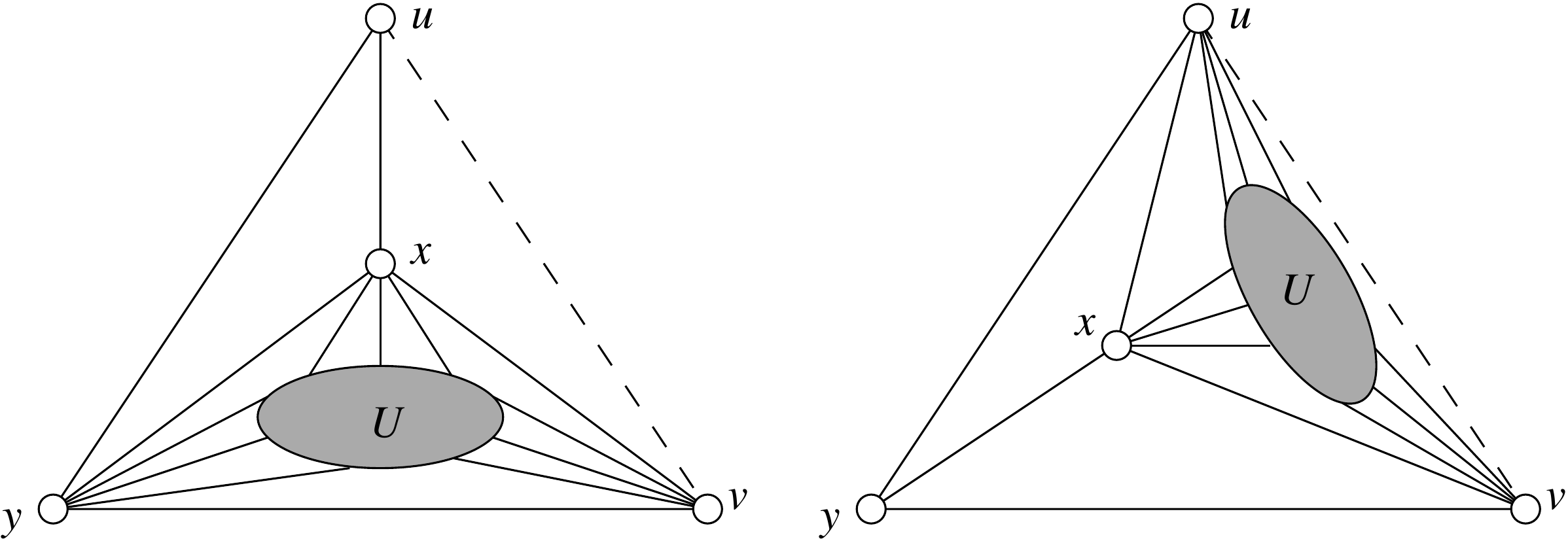}\\
(a)\hspace{0.35\textwidth}(b)
\caption{Illustration for the proof of Lemma~\ref{lem:lam+1-K4}.}
\label{fig:lam+1-K4}
\end{figure}

 Otherwise assume that $f=uvx$ (the case when $f=uvy$ is similar); see
 Fig.~\ref{fig:lam+1-K4}(b). Denote by $X$ the
 subgraph induced by $V(f)\cup U_f$. Due to Lemma~\ref{lem:tight-component}, the subgraph
 of $G-(u,v)$ induced by $U_f\cup\{u,v,x,y\}$ is $(2,3)$-tight. Thus the number of edges
 in $X$ is $2|X|-3$. Furthermore $X-(u,v)$ is $(2,3)$-sparse. Therefore if $X$ is not
 $(2,3)$-tight, then there is subgraph $Y$ of $X$ containing $(u,v)$ such that
 $|E(Y)|\ge 2|V(Y)|-3$.
 If $x$ is not in $Y$, then add the vertex $x$ and the edges $(u,x)$ and $(v,x)$ to it.
 Finally delete from $Y$ the edge $(u,v)$ and add to it the vertex $y$ and the edges $(u,y)$,
 $(v,y)$ and $(x,y)$. Then this graph is not $(2,3)$-sparse but is a subgraph of $G-(u,v)$,
 a contradiction.

We thus assume that for each face $f$ of $H$, the subgraph induced by $V(f)\cup U_f$ is plane
 $(2,3)$-tight. Therefore this graph can be obtained by $E_3$ moves starting with the triangle
 $f$~\cite{FJW04}. Thus using these moves, one can generate $G$ starting from $H$.
\end{proof}

We are now ready to prove Lemma~\ref{lem:lam+1}.

\textit{Proof for Lemma~\ref{lem:lam+1}.}
 Assume for a contradiction that there is a plane Laman-plus-one graph other than $K_4$,
 which cannot be generated by $E_3$ moves starting from $K_4$. Let $G$ be such a graph
 with the minimum number of vertices. Since $G$ is a Laman-plus-one graph, there is an
 edge $e=(u,v)$ of $G$ such that $G-e$ is a Laman ($(2,3)$-tight) graph. We show that
 there is an edge $e'=(x,y)$ such that $u,v\notin\{x,y\}$ which is contractible in $G-e$.
 Thus $H=((G-e)/e')\cup e$ is a Laman-plus-one graph with fewer vertices than $G$ and
 $G$ can be obtained from $H$ by an $E_3$ move. By the minimality of $G$, $H$ (and
 hence $G$) can be generated from $K_4$ by a sequence of $E_3$ moves, a contradiction.

Suppose that there is a triangle $T=xyz$ in $G-e$ such that $u,v\notin\{x,y,z\}$. If
 any edge of $T$ is contractible, then we are done. Otherwise by Lemma~\ref{lem:blocker}
 there exists proper $(2,3)$-tight subgraphs $H_1$, $H_2$, $H_3$ of $G-e$ such that
 $x,y\in V(H_1)$, $z\notin V(H_1)$; $y,z\in V(H_2)$, $x\notin V(H_2)$; and $z,x\in V(H_3)$,
 $y\notin V(H_3)$.  Furthermore for each $i=1,2,3$, there is an edge $e_i\in E(H_i)$ which
 is contractible in $G-e$. By Lemma~\ref{lem:critical}, $V(H_1)\cap V(H_2)=\{y\}$,
 $V(H_2)\cap V(H_3)=\{z\}$, and $V(H_3)\cap V(H_1)=\{x\}$. Therefore there is at least one
 $i\in\{1,2,3\}$, that contains neither $u$ nor $v$. Thus we have a desired contractile edge
 $e_i$ in $G-e$.

We thus assume that each triangle in $G$ contains either $u$ or $v$. Note that no triangle can
 contain both $u$ and $v$ since $G-e$ does not contain the edge $(u,v)$. Since $G-e$ is
 a plane $(2,3)$-tight graph, it contains at least two triangular faces. Let $uxy$ be such a
 triangular face (the case when we consider a triangle containing $v$ is similar). If the edge
 $(x,y)$ is contractible, we are done. Otherwise by Lemma~\ref{lem:blocker}, there is a
 proper $(2,3)$-tight subgraph $H$ of $G$ such that $x,y\in V(H)$ and $u\notin V(H)$.
 If $|V(H)|\ge 4$, then by~\cite{FJW04} there is a contractible edge $e'$ in $H$ that does not
 have $v$ as an endpoint. If $e'$ is not $(x,y)$, then this gives our desired contractible edge.
 Otherwise since each triangle contains either $u$ or $v$, the four vertices $u,v,x,y$ form a
 $K_4$ in $G$. Thus by Lemma~\ref{lem:lam+1-K4}, we can generate $G$ by $E_3$ moves,
 starting from $K_4$, also a contradiction. \qed

Nixon proved the following for $(2,2)$-tight graphs that are not Laman-plus-one~\cite{Nixon-phd}.

\begin{lemma}[Nixon~\cite{Nixon-phd}]
\label{lem:2-2-struct}
 Let $G$ be a $(2,2)$-tight graph with at least one edge. Then either $G$ is a Laman-plus-one graph or there exists a proper $(2,2)$-tight subgraph $H$ of $G$ such that no vertex of $G-H$ is adjacent to more than one vertex in $H$.
\end{lemma}

 For a subgraph $H$ of a graph $G$, let $V(H)$ and $E(H)$ denote the vertex set and the edge
 set of $H$. Denote by $G-H$ the subgraph of $G$ induced by the vertices
 $V(G)\setminus V(H)$. For two subgraphs $H_1$ and $H_2$ of $G$, denote by $H_1\cup H_2$,
 the subgraph $H=(V(H_1)\cup V(H_2), E(H_1)\cup E(H_2))$. Let $E(H_1,H_2)$ be the set of
 edges between the vertices of $H_1$ and $H_2$.

\begin{lemma}
\label{lem:tight-component} Let $G$ be a $(2,k)$-tight graph and let $H$ be a proper
 $(2,k)$-tight subgraph of $G$, for $k>0$. Let $C$ and $D$ be a partition of the vertices
 of $G-H$ such that there is no edge between vertices in C and vertices in $D$. Then both
 the graphs induced by the vertices of $H\cup C$ and $H\cup D$ are also $(2,k)$-tight.
\end{lemma}
\begin{proof}
 Since $G$ and $H$ are both $(2,k)$-tight and $|E(C,D)|=0$,
 $|E(G)| = 2|V(G)|-k = 2|V(H)|+2|V(C)|+2|V(D)|-k = |E(H)|+|E(C)|+|E(H,C)|+|E(D)|+|E[H,D]|$,
 and $|E(H)| = 2|V(H)|-k$.
Again since $H\cup D$ is a subgraph of $G$, it is $(2,k)$-sparse, so
 $|E(H\cup D)| = |E(H)|+|E(D)|+|E[H,D]| \le 2|V(H)|+2|V(D)|-k$.
Thus $|E(H\cup C)| = |E(H)|+|E(C)|+|E[H,C]| \ge 2|V(H)|+2|V(C)|-k$.
Since $H\cup C$ is a subgraph of $G$ and hence is $(2,k)$ sparse,
this implies that it is $(2,k)$-tight. Similarly $H\cup D$ is also $(2,k)$-tight.
\end{proof}

We are now ready to prove Lemma~\ref{lem:moves-generate}.

\begin{proof}[Proof of Lemma~\ref{lem:moves-generate}]
 Part~(\ref{enum:generate-23}), that every plane Laman graph can be generated by $H_1$ and $H_2$ moves starting from a triangle, is already known~\cite{Hen11,Lov82,HOR05}.

 Now, if $G^*$ is the plane dual of a plane Laman graph $G$, then we follow the construction sequence
 of $G$ with $H_1$ and $H_2$ moves in the dual and observe that this gives a construction sequence
 of $G^*$ with $V_2^2$ and $E_3$ moves, starting with three parallel edges.
 This proves~(\ref{enum:generate-23-dual}).

 Let  $G$ be a plane $(2,2)$-tight graph. We prove by induction on $|V(G)|$,
 that $G$ can be generated by $E_3$ and $V_4$ moves, starting with a vertex. If $|V(G)| = 1$,
 this clearly holds.
 Assume that $|V(G)| \geq 2$. If $G$ is Laman-plus-one, it can be obtained from a single vertex by a single $V_4$ move, followed by a number of $E_3$ moves, by
 Lemma~\ref{lem:lam+1}, and the claim follows. Otherwise, by
 Lemma~\ref{lem:2-2-struct}, there exists a proper $(2,2)$-tight subgraph $H$ of $G$ such
 that no vertex of $G-H$ is adjacent to more than one vertex in $H$. Furthermore, since $G$
 is $(2,2)$-tight, $H$ is connected. Since $H$ is a proper subgraph of $G$, assume \WLOG
 that the outer face of $H$ is not vertex-empty in $G$ (otherwise at least one
 internal face is not vertex-empty and a similar reasoning holds). Let $H'$ be the subgraph of
 $G$ consisting of $H$ and all the vertices inside the outer boundary of $H$. By
 Lemma~\ref{lem:tight-component}, $H'$ is also a planar $(2,2)$-tight graph,
 which is a proper subgraph of $G$. Thus by the induction hypothesis, $H'$ can be constructed
 from a single vertex by $E_3$ and $V_4$ moves. Let $\Pi_1$ denote this sequence of these
 two moves. The graph $G'$ obtained from $G$ by merging $H'$ into a
 single vertex is simple and planar $(2,2)$-tight. By the induction hypothesis, $G'$
 can be obtained from a single vertex by a sequence $\Pi_2$ of $E_3$ and $V_4$ moves.
 The sequence $\Pi_2$ followed by the sequence $\Pi_1$ generates
 $G$ from a vertex.
\end{proof}

\begin{lemma}\label{lem:moves-in-CCA}
 Let $G$ be a plane graph with a CCA-representation and $G'$ be a plane graph obtained from $G$ by a $V_4$, $E_3$ or $V_2^2$ move. Then $G'$ admits a CCA-representation as well.
\end{lemma}
\begin{proof}
 Let $G'$ be obtained from $G$ by one move that is either an $V_4$, $E_3$ or $V_2^2$ move. In each case, we show how to locally modify a given CCA-representation of $G$ into a CCA-representation of $G'$. All the cases are similar, so we restrict ourselves to a careful description of the first case only, and provide figures illustrating the remaining cases.

 \begin{figure}[tb]
  \centering
  \includegraphics[scale=0.9]{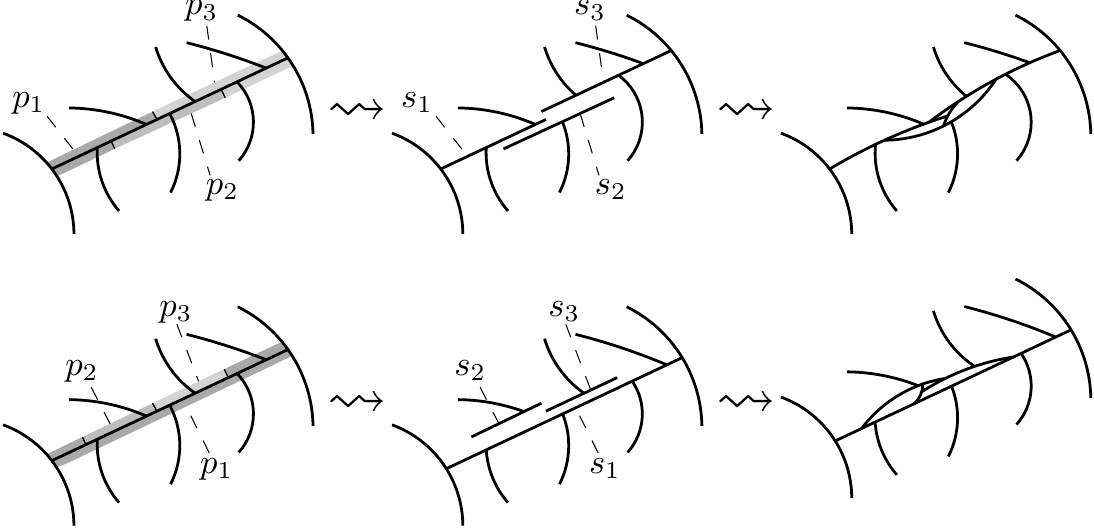}
  \caption{Performing a $V_4$ move in a CCA-representation.}
  \label{fig:V4-CCA}
 \end{figure}

 Let $v$ be the vertex in $G$ and $\{v_1,v_2,v_3,v_4\}$ be the four vertices in $G'$ that replace $v$. Let $S_i = N_{G'}(v_i) \setminus \{v_1,v_2,v_3,v_4\}$, $i = 1,2,3,4$. By definition, we have $S_i \cap S_j = \emptyset$ for $i \neq j$, $S_4 = \emptyset$, $S_1 \cup S_2 \cup S_3 = N_G(v)$ and each of $S_1,S_2,S_3$ forms a subset of $N_G(v)$ that appears consecutively in the clockwise order around $v$ in $G$. Assume \WLOG that the circular arc $c_v$ for $v$ in the given CCA-representation of $G$ is a straight segment. The boundary of $c_v$ can be partitioned into three consecutive pieces $p_1,p_2,p_3$, so that $p_i$ contains exactly the contacts corresponding to vertices in $S_i$, $i = 1,2,3$; see Fig.~\ref{fig:V4-CCA}.

From the pieces $p_1,p_2,p_3$, we define straight segments $s_1, s_2, s_3$ parallel to $c_v$ so that each $s_i$ intersects exactly the circular arcs for vertices in $S_i$, $i = 1,2,3$. Then, each $s_i$ is ``curved'' into a circular arc, so that $s_1,s_2,s_3$ form a triangle with one free endpoint on the inside. We add a fourth circular arc for $v_4$ in the triangle, containing the free endpoint in its interior and touching the other two circular arcs with its two endpoints; see Fig.~\ref{fig:V4-CCA}.

 \begin{figure}[tb]
  \centering
  \includegraphics[scale=0.9]{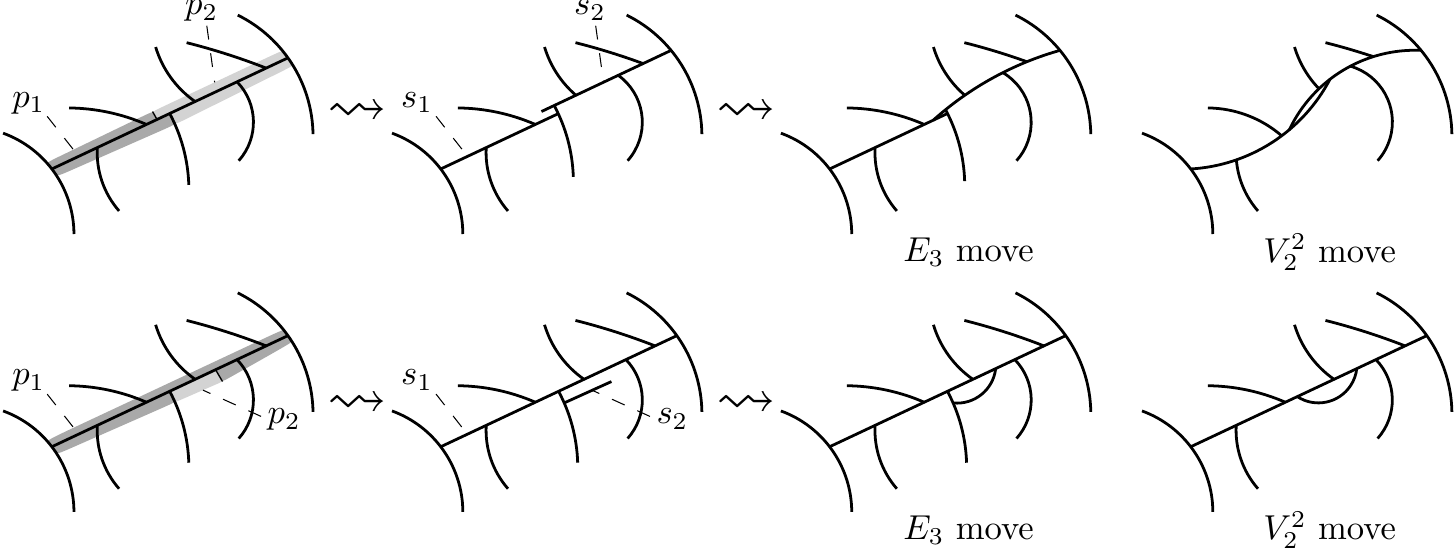}
  \caption{Performing an $E_3$ move and $V_2^2$ move in a CCA-representation.}
  \label{fig:other-CCA}
 \end{figure}

 The cases for $E_3$ move or $V_2^2$ are similar; the only difference is that we define only two sets $S_1,S_2$ (with $S_1 \cap S_2 = \{u\}$ in the case of an $E_3$ move) and consequently only two pieces $p_1,p_2$ and two straight segments $s_1,s_2$; see Fig.~\ref{fig:other-CCA}.
\end{proof}

Finally, we prove the main theorem of this section.

\begin{proof}[Proof of Theorem~\ref{thm:CCA}, Cases~1 and~2]
 Let $G$ be a plane graph. We shall show that $G$ admits a CCA-representation, provided it is $(2,2)$-sparse or a co-Laman graph.
 \begin{description}
  \item[Case~1, $G$ is $(2,2)$-tight.] By Lemma~\ref{lem:moves-generate}, $G$ can be generated by $E_3$ and $V_4$ moves, starting from a graph $G_0$ consisting of a single vertex. Since $G_0$ admits a CCA-representation (Fig.~\ref{fig:base-case-CCA-A}), by Lemma~\ref{lem:moves-in-CCA},  $G$ also admits a CCA-representation.

  \item[Case~2, $G$ is co-Laman.] By Lemma~\ref{lem:moves-generate}, $G$ can be generated by $E_3$ and $V_2^2$ moves, starting from a graph $G_0$ with two vertices and three parallel edges. Since $G_0$ admits a CCA-representation (Fig.~\ref{fig:base-case-CCA-B}), by Lemma~\ref{lem:moves-in-CCA}, $G$ also admits a CCA-representation.
  \end{description}
\vspace{-3ex}
\end{proof}

\begin{figure}[tb]
 \centering
 \subfigure[]{
    \includegraphics[scale=0.6,page=1]{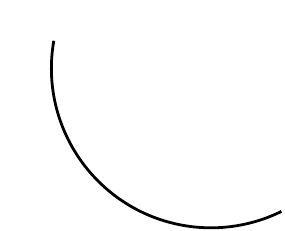}
    \label{fig:base-case-CCA-A}
 }
 ~~
 \subfigure[]{
    \includegraphics[scale=0.6,page=2]{base-case-CCA}
    \label{fig:base-case-CCA-B}
 }
 ~~~~~~~~~~~~~
 \subfigure[]{
    \includegraphics[scale=0.75]{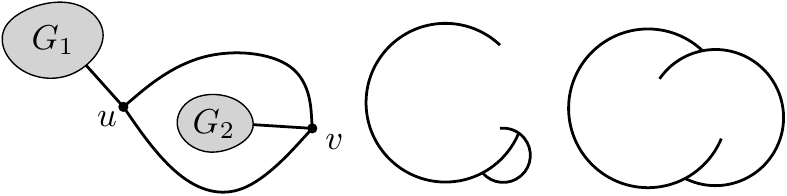}
    \label{fig:20-counter}
 }
 \caption{The base case $G_0$ for (a)~$(2,2)$-tight graphs and (b)~co-Laman graphs.
  (c)~A plane $(2,0)$-tight graph that does not admit a CCA-representation.}
\end{figure}

\section{Good Orientations and One-Sided Representations}
\label{sec:good}

In order to prove Theorem~\ref{thm:CCA}, we need an additional concept of a plane graph $G$. An orientation of $G$ is called \df{good} if, for every vertex $v$ of $G$, all the outgoing edges (equivalently incoming edges) incident to $v$ are consecutive in the circular ordering of the edges around~$v$. A CCA-representation is called \df{one-sided} if, for each arc $a$, the endpoints of other arcs that touch $a$ all do so on one side of~$a$. This analogous to the concept of one-sided segment contact representations~\cite{Hli98,EppMumSpe-SJC-12}. A CCA-representation is  \df{interior-disjoint} if each arc has nonzero curvature and the interior of the convex hull of each arc is disjoint from all the other arcs.

\begin{lemma}\label{lem:good-to-CCA}
 A simple plane graph with a good $2^-$-orientation has an interior-disjoint CCA-representation.
\end{lemma}
\begin{proof}
Consider a plane graph $G$ with a good $2^-$-orientation. As with any plane graph, $G$ has a contact representation with disks~\cite{Koebe36}. For each vertex $v$ of $G$ with $\outdeg(v) = 2$, the two outgoing edges of $v$ define two points, $p$ and $p'$, on the circle $C(v)$ representing $v$. If $\outdeg(v)= 1$, the outgoing edge defines $p$ and we choose $p' \in C(v)$ very close to it. If $\outdeg(v) = 0$ we choose $p$ and $p'$ distinct from all contacts of $C(v)$ and close to each other. In both cases the two points, $p$ and $p'$, split $C(v)$ into two circular arcs. Since the $2^-$-orientation is good, one of these two arcs contains none of the contacts of $C(v)$ with other disks. We represent each vertex $v$ by the other circular arc defined by $C(v),p$ and $p'$, which contains all the contacts of $C(v)$; see Fig.~\ref{fig:good-to-CCA}.
\end{proof}

\begin{lemma}\label{lem:4-reg}
Every plane graph $G$ with maximum degree $4$ has a good $2^-$-orientation.
\end{lemma}
\begin{proof}
 First note that vertices of degree strictly less than $4$ are harmless, as long as they have at most two outgoing edges: their outgoing edges (if any) cannot be non-consecutive. In order to find a $2^-$-orientation of $G$ in which every vertex of degree $4$ has consecutive outgoing edges, we define a number of walks in $G$. Start with any edge and define a walk so that, when entering some vertex $v$ of degree $4$ via edge $e$ the walk always continues with the edge $e'$ that lies opposite of $e$ at $v$. At a degree-3 vertex that has not already been made part of one of these walks, continue the walk with an arbitrary incident edge, and otherwise stop. Orienting every edge in this walk consistently and starting another iteration with any so-far unoriented edge (if any exists), eventually gives a good $2^-$-orientation of $G$.
\end{proof}

 Not every planar $(2,0)$-sparse graph has a good orientation; a counterexample is easy to construct by adding sufficiently many degree-$2$ vertices. Moreover, there is a counterexample with minimum degree $3$; see Fig.~\ref{fig:counter}. Indeed, the bold subgraph (the one induced
 by the black vertices) has five edges and four vertices. Thus,  in any $2^-$-orientation, at least one black vertex has two bold outgoing edges. At this vertex, all the light edges must be incoming, necessarily breaking up its two outgoing edges.

\begin{figure}[tb]
 \centering
 \subfigure[]{
  \includegraphics[scale=0.9]{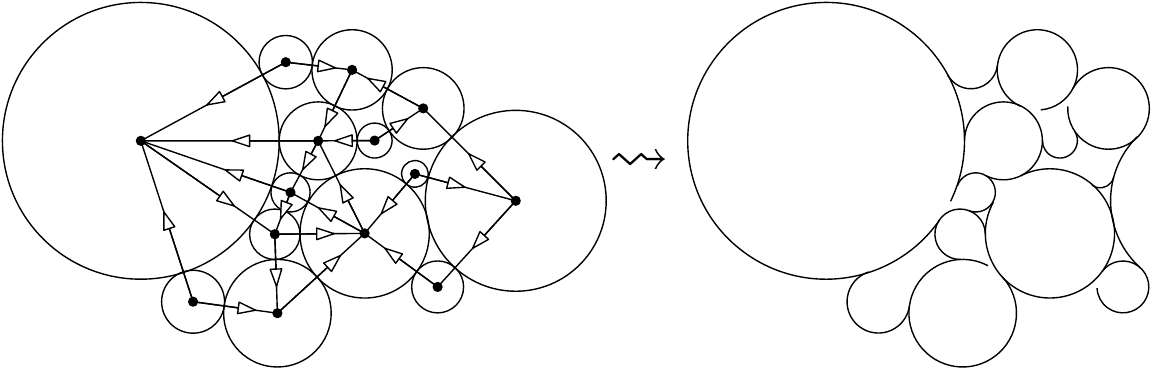}
  \label{fig:good-to-CCA}
 }
 \hfill
 \subfigure[]{
  \includegraphics[width=0.2\textwidth]{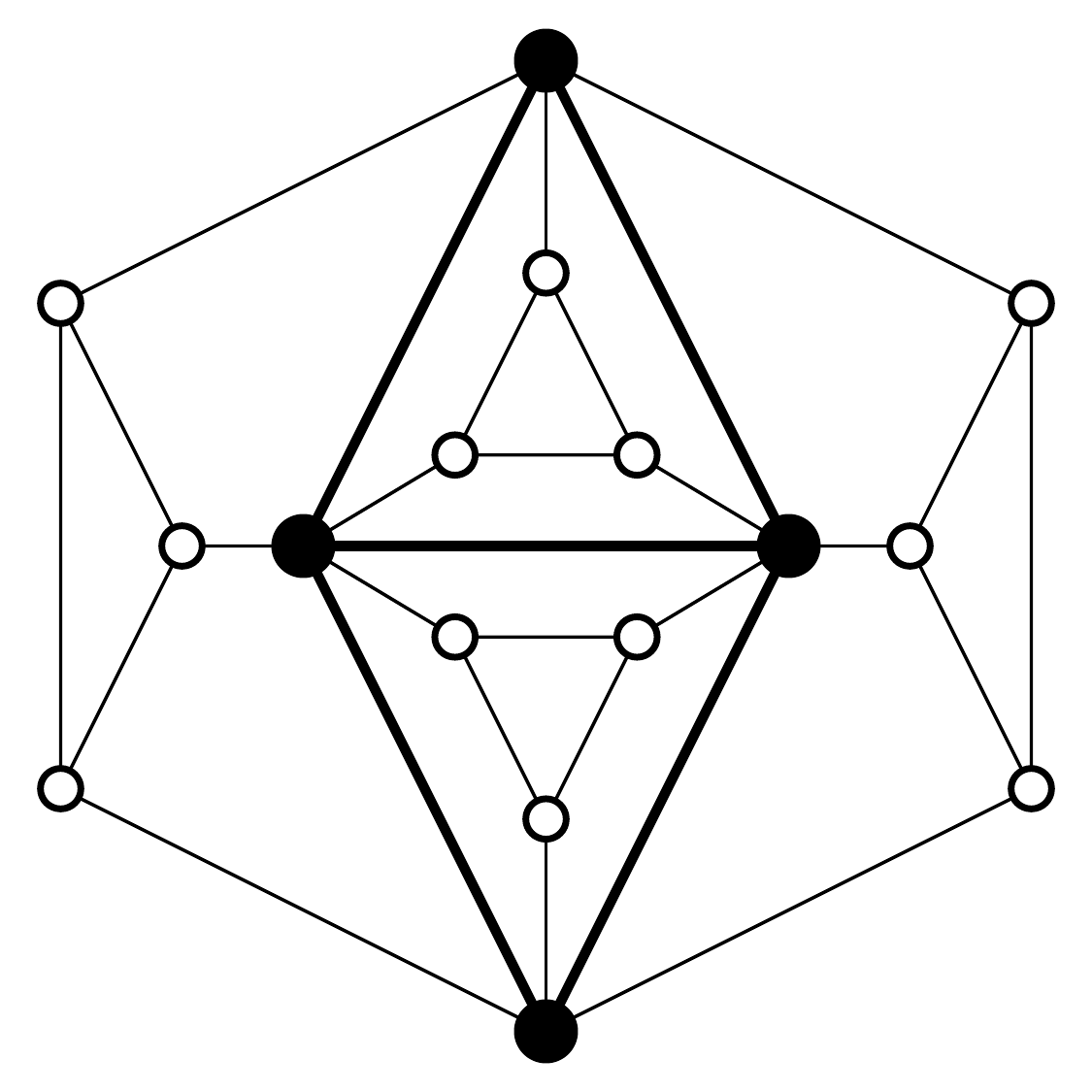}
  \label{fig:counter}
 }
 \caption{(a)~From a contact representation with disks and a good $2^-$-orientation to a CCA-representation.
  (b)~A planar Laman graph with the minimum degree 3, that has no good $2$-orientation.}
\end{figure}

\begin{lemma}
\label{lem:good-equivalents}
For a simple plane graph $G$, the following three statements are equivalent.
~\\[-4ex]
\begin{enumerate}[\rm (1)]
\item $G$ has a good $2^-$-orientation
\item $G$ has a one-sided CCA-representation
\item $G$ has an interior-disjoint CCA-representation
\end{enumerate}
\end{lemma}

\begin{proof}
The implication (3) $\Rightarrow$ (2) is obvious (every interior-disjoint representation is one-sided) and the (1) $\Rightarrow$ (3) is Lemma~\ref{lem:good-to-CCA}.
It remains to prove that every graph with a one-sided CCA-representation has a good $2^-$-orientation. But each vertex of the $2^-$-orientation derived from a CCA-representation has incoming edges on the two sides of the corresponding arc separated by outgoing edges at the two arc endpoints; in a one-sided representation, one set of incoming edges is empty and cannot separate the outgoing edges.
\end{proof}

\begin{proof}[Proof of Theorem~\ref{thm:CCA}, Cases~3, 4]
We first show that a plane graph $G$ with maximum degree 4 admits a CCA-representation.
 By Lemma~\ref{lem:4-reg}, $G$ admits a good $2^-$-orientation; and hence by
 Lemma~\ref{lem:good-to-CCA} $G$ admits a CCA-representation.
 This completes Case 3.

A plane $(2,0)$-tight multigraph with no CCA-representation is shown in Fig.~\ref{fig:20-counter}. It has two vertices $u, v$ joined by two parallel edges $e,e'$, and two plane $(2,0)$-tight subgraphs $G_1$ and $G_2$. $G_1$ lies in the unbounded region and $G_2$ lies in the bounded region defined by $e,e'$, and $u$ and $v$ are connected by an edge to a vertex in $G_1$ and $G_2$, respectively. $G$ is plane and $(2,0)$-tight and admits no CCA-representation since two touching circular arcs have their free ends either both in the bounded or both in the unbounded region defined by the closed created curve (Fig.~\ref{fig:20-counter}). Note that whether every \emph{planar} $(2,0)$-tight multigraph has a plane embedding that admits a CCA-representation is an open question.
 \end{proof}

As noted in the introduction, Case~3 of the proof of Theorem~\ref{thm:CCA} always constructs an interior-disjoint CCA-representation for graphs of maximum degree~4. As we now show, finding such representations without this degree constraint is hard.

\begin{figure}[t]
 \centering
 \subfigure[wire]{
  \centering\includegraphics[width=.075\textwidth]{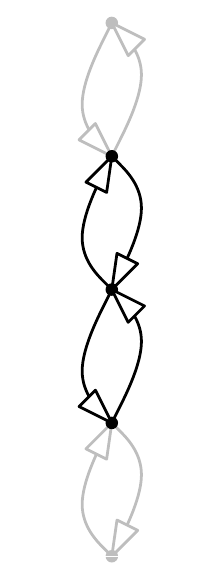}
 }
 \subfigure[splitter]{
  \includegraphics[width=0.19\textwidth]{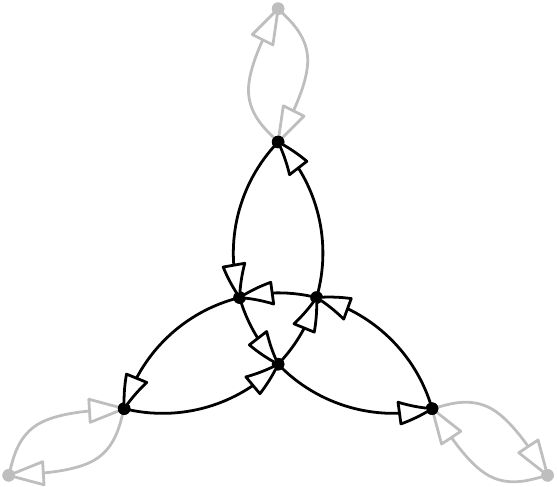}
 }
 \subfigure[clause]{
  \includegraphics[width=0.19\textwidth]{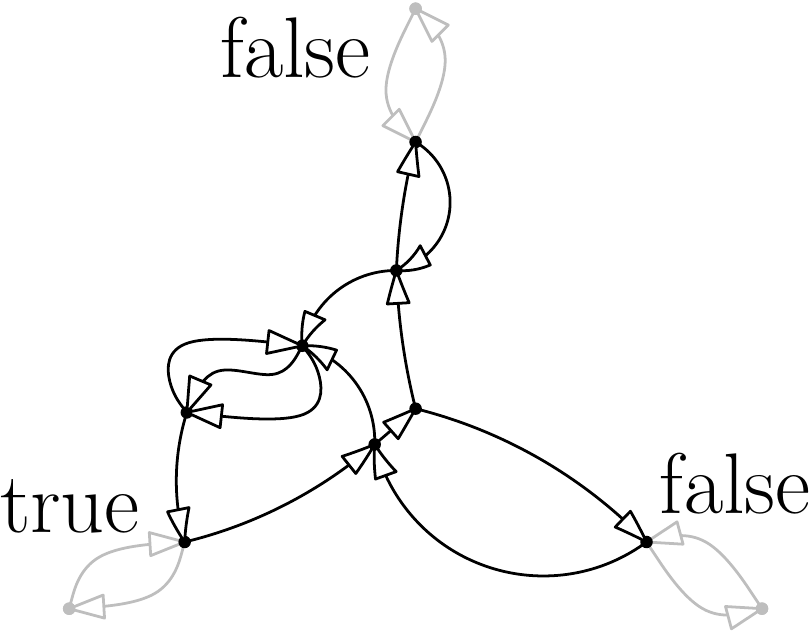}

 }
 \subfigure[]{
  \includegraphics[width=0.10\textwidth]{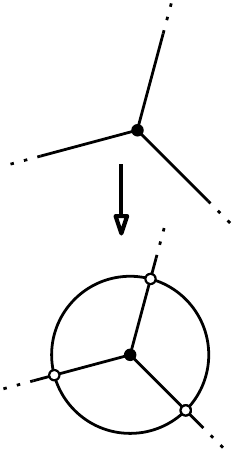}
  \label{fig:simplify-good-a}
}
 \subfigure[]{
  \includegraphics[width=0.28\textwidth]{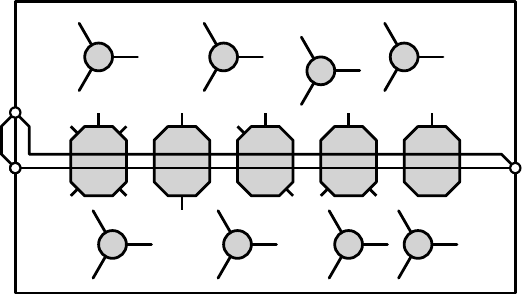}
  \label{fig:simplify-good-b}
 }
 \caption{(a)--(c)~The three gadgets used in the reduction. The gray edges encode how adjacent wire gadgets are connected. (d)~Removing multiple edges. (e)~Piercing the variable gadgets (octagons) to get the 3-connectedness.}
\label{fig:redux2multigraph}
\end{figure}

\begin{proof}[Proof of Theorem~\ref{thm:empty-hull-NPC}]
We first prove the theorem for plane multi-graphs. By Lemma~\ref{lem:good-equivalents} it is equivalent
 to prove that finding a good representation of a $(2,0)$-tight simple plane graph is NP-complete.
We reduce from the known NP-complete problem {\sf Positive Planar 1-in-3SAT}~\cite{MR08} of deciding the satisfiability of a 3SAT formula with the additional restrictions that: (1) the formula contains only positive variables; (2) we seek a truth assignment such that each
clause has exactly one true variable; and (3) the graph whose vertices are variables and clauses and whose edges are variable-clause incidences has a planar drawing in which all variables are arranged on a straight line and no edges cross this line.

Our reduction uses \emph{wire}, \emph{splitter}, and \emph{clause} gadgets (Fig.~\ref{fig:redux2multigraph}) based on the principle that, for a good $2$-orientation at a degree-$4$ vertex, each incoming edge must be opposite an outgoing edge and vice versa. Each variable of the 3SAT formula is replaced by
 a wire gadget that is closed to a circle with doubled edges. There are two good $2$-orientations of this doubled circle, the choice of which encodes the truth value of that variable. A splitter gadget with a short
 piece of wire propagate this signal to the clause gadgets, which represent 3SAT clauses.
Due to the degree-$4$ vertices, there is only one good $2$-orientation for the splitter
 that extends a ``wire signal''.
The degree-$3$ vertex in the clause gadget verifies that exactly two of the attached wires
 carry a false signal. The remaining edges and vertices in the clause gadget
ensure that a completion to a good $2$-orientation is possible in all cases.


To finish the proof, we convert our $(2,0)$-tight multigraph into a $(2,0)$-tight simple planar graph,
preserving the existence of a good $2$-orientation. To accomplish this, we surround each
 vertex by a chain of edges and vertices (Fig.~\ref{fig:simplify-good-a}) and augment the graph to make it $3$-connected (Fig.~\ref{fig:simplify-good-b}),
ensuring that our chosen combinatorial embedding is unique up to reflection. To do this we add a triangle with doubled edges that will enclose the previous construction.
We route two non-parallel edges of the triangle so they cut through edges of all variable gadgets,
and repair the crossings by adding new vertices at the intersections.
We claim that  every vertex has three vertex-disjoint paths to vertices on the boundary,
 which implies the $3$-connectedness. If we follow a wire in one direction we can do this
 in two possible vertex-disjoint ways (after the removal of the multiple edges). Hence, if we consider
 a vertex of a wire gadget we can reach the variable gadgets with three vertex-disjoint paths;
 two in one direction and one taking a detour via a clause gadget and another wire.
From the variable gadgets we can now direct the three paths to the boundary. Again, since the
 variables are built from wires, we can select two paths along one ``direction'' and use the other
 direction for the third path.
The argument is similar for the other vertices. Note that the augmented graph still has one
 multiple edge, which we can remove with the same construction. The resulting graph has maximum vertex degree~$5$,  in its clause gadgets.
\end{proof}

\section{Contact Representations with Wedges}
\label{sec:wedges}
\label{sec:wedge}

A \df{wedge} is a polyline segment with at most one bend (and hence a sequence of two circular arcs). Here we show that any plane $(2,0)$-sparse
 graph admits a contact representation with wedges.
 We also consider representation with constrained wedges.

\begin{figure}[htbp]
 \centering
 \includegraphics[width=0.22\textwidth]{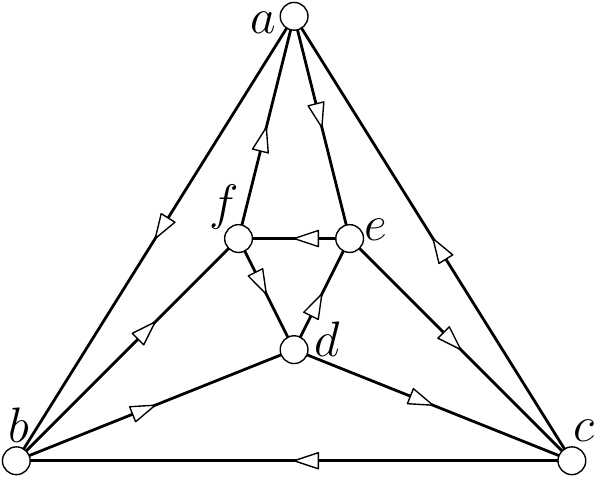}
~~~~~~~~~~~~~~
 \includegraphics[width=0.22\textwidth]{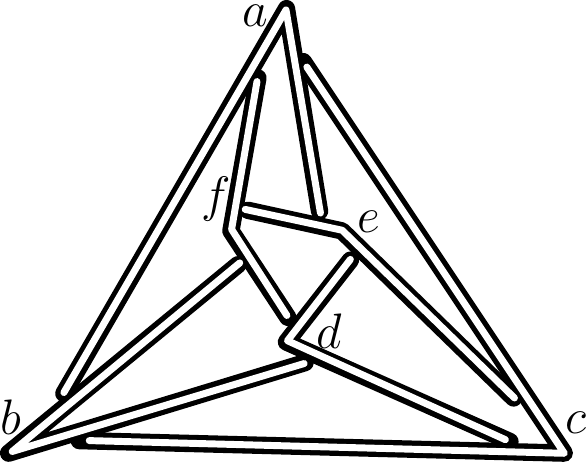}
 \caption{(a)~A straight-line drawing of a $(2,0)$-sparse graph $G$ and a $2$-orientation of $G$,
 (b)~a~contact representation of $G$ with wedges.}
 \label{fig:wedge}
\end{figure}

\begin{backInTime}{thm-wedge}
\begin{theorem}
Every plane $(2,0)$-sparse graph has a contact representation where each vertex is represented by a wedge.
\end{theorem}
\end{backInTime}


\begin{proof} A plane $(2,0)$-sparse graph $G$ has a $2^-$-orientation~\cite{FM01,BF12,NW61}.
 Consider a straight-line drawing of $G$. For each vertex
 $v$, the wedge for $v$ is the union of the straight-line segments representing the outgoing edges from $v$.
Here all the contacts representing the incoming edges for a vertex is at the bend-point of the wedge,
 but a small perturbation of the representation is sufficient to get rid of this degeneracy; see
 Fig.~\ref{fig:wedge}.
\end{proof}

We now consider contact representation with \df{right wedges}, where the angle at the corner of
 each wedge is $90^\circ$.

\begin{figure}[htbp]
 \centering
 \includegraphics[width=0.25\textwidth]{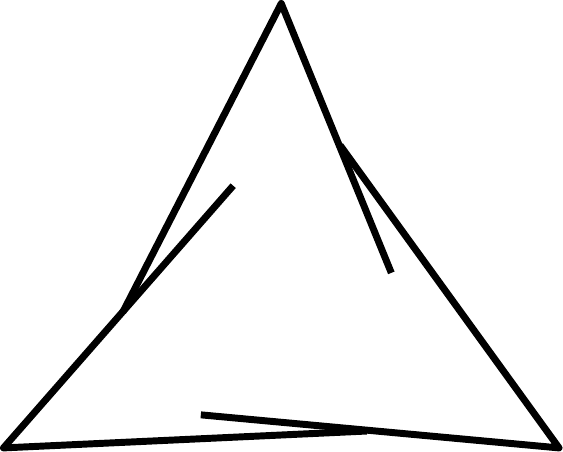}
 \caption{
A wedge-contact representation bounded by three wedges.}
 \label{fig:wedge-app}
\end{figure}

\begin{lemma}\label{lem:no-90}
 Let $\Gamma$ be a wedge-contact representation of a $(2,0)$-tight graph such that the outer boundary of $\Gamma$ is bounded by exactly three wedges. Then $\Gamma$ cannot be modified to a topologically equivalent representation where all the three wedges on the outer boundary are right wedges.
\end{lemma}

\begin{proof}
 Since the underlying graph is a $(2,0)$-tight graph, all the angles in the outer boundary of $\Gamma$ other than those at the three corners of the wedges are reflex angles. Thus the outer boundary of $\Gamma$ is a pseudo-triangle; see Fig.~\ref{fig:wedge-app}. If all three wedges are right angle, we have a pseudo-triangle with the summation of the three convex angles more than $180^\circ$; which is geometrically infeasible.
\end{proof}

Lemma~\ref{lem:no-90} immediately implies that if a plane $(2,0)$-tight graph has a triangular outer face, then it has no contact representation with right wedges. In particular, the octahedron graph is a $(2,0)$-tight graph and all its faces are triangles. Thus in every planar embedding of the octahedron, the outer face is a triangle, so the octahedron has no contact representation with right wedges.

\section{Conclusion and Open Questions}

We presented new results about contact representations of graphs with circular arcs. Although every graph with such a contact representation is planar and
 $(2,0)$-sparse, we provided a $(2,0)$-tight plane multigraph that does not admit such a
 representation. On the other hand, we identified
 several subclasses of plane $(2,0)$-sparse graphs that have CCA representations.
 The natural question remains open: does every simple planar $(2, 0)$-sparse graph have a
 circular-arc contact representation, if we allow changing the embedding?

A circular-arc contact representation $\Gamma$ for a $(2,0)$-tight graph $G$ defines a $3$-regular
 \emph{skeleton graph}~\cite{AnS13}, where the points of contact are vertices and arcs between contacts are edges (Fig.~\ref{fig:teaser}). Each vertex of $G$ with degree $d\ge 3$ corresponds to a path of $d-2$ vertices
 in $\Gamma$ along one circular arc. Thus one possible way to find a circular-arc contact representation
 for a plane $(2,0)$-tight graph $G$ is to find a 3-regular graph by splitting each vertex of degree
 $d>3$ into $d-2$ degree-3 vertices (each choice of splitting corresponds to a different
 $2^-$-orientation), and then align the path associated with each vertex of $G$ into a circular arc.
 Given a 3-regular graph with a path-cover, finding a representation with each path aligned as a circular arc is related to the stretchability question~\cite{FM07barycentric},
 which is still open.

We also showed that every plane $(2,0)$-sparse graph has a contact
representation with polyline segments with a single bend (wedges).
In this context, several
questions seem interesting: does every $(2,0)$-sparse graph admit a contact representation with
equilateral wedges (that is, wedges with equal-length segments)? Can we bound the smallest
size of the angle at the corner of the wedges (to say, $45^\circ$)?

{

\end{document}